\documentclass[aps,prl,reprint,amsmath,amssymb,nofootinbib,longbibliography]{revtex4-2}

\usepackage{amsthm,bm,braket,graphicx}
\usepackage[colorlinks=true,linkcolor=blue,citecolor=blue]{hyperref}

\newtheorem{axiom}{Axiom}
\newtheorem{theorem}{Theorem}

\newtheorem{corollary}{Corollary}
\theoremstyle{definition}

\newcommand{\Sgroup}{\mathcal{S}}
\newcommand{\Ngroup}{\mathcal{N}}
\newcommand{\Code}{\mathcal{C}}
\newcommand{\Wop}{\hat{W}}
\newcommand{\Urel}{U_{\mathrm{rel}}}
\newcommand{\CU}{\mathcal{U}}
\newcommand{\CV}{\mathcal{V}}          % <-- was missing in v2/v3
\newcommand{\CE}{\mathcal{E}}
\newcommand{\Tbar}{\bar{T}}

\newcommand{\Id}{\mathbb{I}}
\newcommand{\id}{\mathrm{id}}
\newcommand{\Sbar}{\bar S}
\newcommand{\Zbar}{\bar{Z}}
\newcommand{\Xbar}{\bar{X}}
\newcommand{\Tr}{\mathrm{Tr}}
\newcommand{\Mbar}{\overline{M}}

\begin{document}
%
% TITLE SHORTLIST (recommendation first):
%   Quantum Error-Corrected Memories Keep Proper Time  <-- CHOSEN
%     'Quantum' is load-bearing: ECC memory is standard classical-computing
%     terminology, so the shorter form is ambiguous in a reference list.
%     Names Thm 1, which is the load-bearing citation for downstream work
%     (encoding-superposition protocols cite the license, not the gate).
%   Proper Time in Error-Corrected Quantum Memories    (maximally neutral fallback)
%   Proper Time as a Logical Gate                      (names Cor. 6; narrower)
%   Quantum Error Correction Is Blind to Proper Time   (sharper, but the negation
%     can be misread as "you cannot measure proper time with logical qubits")
%   A Worldline Excursion Is a Logical Gate            (most distinctive, least findable)
\title{Quantum Error-Corrected Memories Keep Proper Time}

\author{Sebastian P.~Kish}
\affiliation{Centre for Quantum Software and Information, University of Technology Sydney, NSW 2007, Australia}

\author{Simon J.~Devitt}
\affiliation{Centre for Quantum Software and Information, University of Technology Sydney, NSW 2007, Australia}

\date{\today}

\begin{abstract}
What does quantum error correction do to relativistic proper time? We show that when two error-corrected quantum memories follow distinct worldlines and are controlled by local references, their local correction histories factor \emph{exactly} from the proper-time difference between them. Error correction therefore does not erase time dilation: when the memories are compared, the accumulated difference appears as a relational phase rotation. During locally referenced operation this rotation is invisible to every local error syndrome; at discrete phases admitted by a code, the final relational evolution instead returns exactly as a deterministic logical gate, including an $S$-type Clifford gate at a $\pi/2$ excursion. These results connect fault-tolerant quantum information with relativistic clocks and establish protected quantum memories as a platform for preserving proper-time information and for long-coherence tests of quantum dynamics involving relativistic proper time.
\end{abstract}

\maketitle
\emph{Introduction.}---Proper time is routinely read from physical qubits: since Pound and Rebka~\cite{PoundRebka60} the redshift has been resolved across $33$~cm~\cite{Chou10}, across the $450$~m of the Tokyo Skytree~\cite{Takamoto20}, and across a millimetre of one atomic sample~\cite{Bothwell22}. A parallel line places the clock itself in superposition, so that proper time becomes which-path information rather than a parameter~\cite{ZCPB11,Pikovski15,Loriani19,Roura20}, with recent ion-clock proposals extending this regime~\cite{Sorci26}; entangled clocks distributed over a quantum network sharpen the test~\cite{Borregaard2025}. Distributed atomic processors and quantum memories have likewise been proposed as probes of differential proper time~\cite{Covey2025,Barzel2024,Gundogan2026}, while long-baseline single-photon interferometry is approaching the sensitivity required for laboratory gravitational-phase tests~\cite{Yu2026}. Throughout, the record is a relative phase carried by a bare transition or unencoded memory, so such comparisons remain bounded by physical coherence.

Extending such comparisons beyond the physical-memory coherence time naturally suggests an error-corrected memory. But the phase a clock accumulates is itself a collective $Z$ rotation, of the kind an error-correcting code may be designed to suppress. This raises a question the physical-qubit experiments cannot pose: what does error correction do to proper time, and is the signal corrected away with the noise? For clock synchronization between parties at rest, Preskill answered it, and negatively~\cite{Preskill00}: the natural precession is what a phase-error-correcting code is built to resist, so stabilizing the pair freezes it, and what survives is an offset the parties could have read from their own clocks.

This Letter considers the distinct-worldline setting: each oscillator shares a worldline with the qubits it drives, so neither register is detuned from its own reference and there is nothing to freeze. What accumulates is a relation between frames. Under locally referenced control, error correction cannot touch this relation, exactly and independently of the code; the same blindness prevents a protected code from registering it at first order. Thus the relation is a syndrome-invisible relational error at small $\delta\theta$, while at discrete phases admitted by a code it is a deterministic logical gate.

\emph{Setup.}---In the momentarily comoving rest frame a clock transition is generated by $H=\tfrac{\hbar\omega_0}{2}Z$, with $\hbar$ the reduced Planck constant, $Z$ the Pauli-$Z$ operator, and $\omega_0$ the \emph{proper} transition frequency, the same on every worldline by local position invariance. The state evolves in the qubit's own proper time,
\begin{equation}
  i\hbar\,\partial_\tau\ket{\varphi}=H\ket{\varphi}
  \;\Longrightarrow\;
  \ket{\varphi(\tau)}=R_Z(\omega_0\tau)\ket{\varphi(0)},
  \label{eq:precess}
\end{equation}
$R_Z(\theta)\equiv e^{-i\theta Z/2}$. 

Parametrise two worldlines by a common coordinate time $t$ via $\tau_M(t)=\int_0^t\!\sqrt{-g_{\mu\nu}\dot x_M^\mu\dot x_M^\nu}\,dt'/c$, $M\in\{A,B\}$, where $g_{\mu\nu}$ is the spacetime metric, a dot denotes $d/dt$, and $c$ is the speed of light. For a static weak field, with local gravitational acceleration $g$, height $h_M$ and speed $v_M$, $d\tau_M/dt=1+gh_M/c^2-v_M^2/2c^2+\cdots$. Define $\Delta\tau\equiv\tau_B-\tau_A$ and $\Delta h\equiv h_B-h_A$. Only $\Delta\tau$ enters below; potential and motion contribute alike, and the quantitative examples in this Letter use the former. Let each memory be $n$ physical qubits sharing worldline $\gamma_M$ and one oscillator, ions in a single trap being the concrete case.
% SPK: intra-register vertical-extent footnote cut. Recoverable:
% \footnote{A register shares a worldline only to the accuracy with which its qubits do. A vertically oriented millimetre-scale string spans $5.5\times10^{-9}$~rad\,s$^{-1}$ of proper-frequency difference end to end, or $7\times10^{-6}$~rad over the $T_\pi=1267$~s hold at $\Delta h=450$~m, an infidelity $\sim10^{-11}$ before correction. Unlike $\Urel$ this residue is \emph{differential} and therefore detectable, being ordinary single-qubit dephasing, and horizontal orientation removes it entirely. We take $\theta_M$ single-valued throughout.}
With $\theta_M(t)\equiv\omega_0\tau_M(t)$ and $\Wop\equiv\sum_{i=1}^n Z_i$, with $\Wop_M$ its copy on register $M$, Eq.~\eqref{eq:precess} gives
\begin{equation}
  U_{\mathrm{lab}}(t)=U_A(t)\otimes U_B(t),
  \quad
  U_M(t)=e^{-i\theta_M(t)\Wop_M/2},
  \label{eq:ulab}
\end{equation}
and $\delta\theta(t)\equiv\theta_B-\theta_A=\omega_0\Delta\tau(t)$. Over a hold of duration $T$, both $\theta_M$ are of order $\omega_0T\sim10^{12}$~rad; only the difference is physical. That the common mode drops out is not obvious, since the logical states of a general code are not $\Wop$ eigenstates and $e^{-i\theta\Wop/2}$ is then not a global phase. For a density operator $\rho$, write $\CU_M(t)(\rho)=U_M\rho U_M^\dagger$ and $\CU_{\mathrm{lab}}=\CU_A\otimes\CU_B$. We use $\Id_M$ and $\id_M$ for the identity operator and identity channel on register $M$, respectively.

\emph{What a local operation is.}---An operation has an equatorial axis, set by an oscillator. The $Z$ eigenbasis is fixed by the physics, but the phase of $\ket{1}$ relative to $\ket{0}$ is a free relabelling by an arbitrary angle $\chi$, under which $X\mapsto\cos\chi\,X+\sin\chi\,Y$: the Pauli $X$ is an operator together with a phase reference, and one must say which. That two parties measuring $X$ must first agree that reference was noted in Ref.~\cite{Preskill00}.

\begin{axiom}[Local control reference]\label{ax:comoving}
Control operations on register $M$, comprising encoding, syndrome extraction, correction and readout, are generated on $\gamma_M$ and referenced to an oscillator on $\gamma_M$ phase-locked to the local clock transition. A \emph{local control operation on $M$ at coordinate time $t$} is therefore
\begin{equation}
  \CE^{(t)}_M=\CU_M(t)\circ\CE\circ\CU_M(t)^{-1},
  \label{eq:axiom}
\end{equation}
with $\CE$ any completely positive trace-preserving (CPTP) map supported on $M$'s register and on ancillas co-located on $\gamma_M$, extended by $\id_{\Mbar}$ to the joint space, where $\bar M$ denotes the other register.
\end{axiom}

The axiom is about apparatus, under the experimenter's control: the synthesiser on $\gamma_M$ is locked to the qubits it drives, which \emph{defines} $M$'s rotating frame, and any static offset is common to the register and absorbed into $\theta_M$. Noise is not, and the following extends the axiom to it.

Standard longitudinal dephasing is frame-covariant trivially. Appendix~\ref{app:rwa} shows that, to leading rotating-wave order, near-resonant transverse semiclassical noise likewise has the covariance form of Axiom~\ref{ax:comoving}; the treatment of the residual event offset is given in Appendix~\ref{app:offsets}. Thus Theorem~\ref{thm:reduction} applies to histories of local frame-covariant operations and errors satisfying Axiom~\ref{ax:comoving}. The resulting frame reduction is illustrated in Fig.~\ref{fig:frames}.

\begin{figure}[t]
  \centering
  \includegraphics[width=\columnwidth]{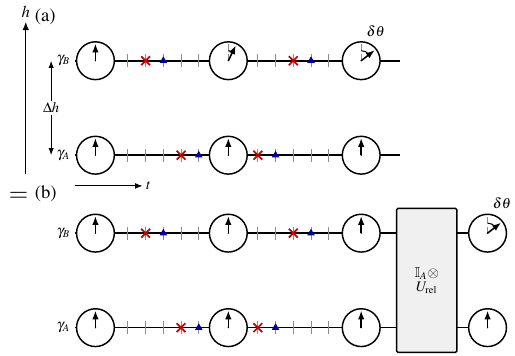}
  \caption{(a) Locally corrected memories at heights differing by $\Delta h$ accumulate only the relative phase $\delta\theta(t)=\omega_0\Delta\tau(t)$. Ticks are syndrome rounds; crosses and triangles are local errors and corrections. (b) Theorem~\ref{thm:reduction}: the same local circuit, rotation-free, followed by $\Id_A\otimes\Urel$.}
  \label{fig:frames}
\end{figure}

\emph{Frame reduction.}---
\begin{theorem}\label{thm:reduction}
Let $\CE^{(t_1)}_{M_1},\dots,\CE^{(t_m)}_{M_m}$ be local operations in the sense of Axiom~\ref{ax:comoving}, $M_a\in\{A,B\}$ for $a=1,\ldots,m$, $0\le t_1\le\cdots\le t_m\le T$, interleaved with the free evolution~\eqref{eq:ulab}, and let $\Lambda$ be the resulting lab-frame channel. Let readout at $T$ be performed with $A$'s oscillator, i.e.\ preceded by $\CV(T)^{-1}$, where $V(T)=R_Z(\theta_A(T))^{\otimes2n}$ on every qubit of both registers and $\CV(T)(\rho)=V(T)\rho V(T)^\dagger$. Then
\begin{equation}
  \CV(T)^{-1}\!\circ\Lambda
  =\big(\id_A\otimes\CU_{\mathrm{rel}}\big)\circ\CE_{M_m}\!\circ\cdots\circ\CE_{M_1},
  \label{eq:reduction}
\end{equation}
where $\Urel=e^{-i\delta\theta(T)\Wop_B/2}$, $\CU_{\mathrm{rel}}(\rho)=\Urel\rho\Urel^\dagger$, $\id_A$ is the identity channel on register $A$, and the $\CE_{M_a}$ are the bare, frame-independent maps of Eq.~\eqref{eq:axiom} with all conjugations stripped.
\end{theorem}

\begin{proof}
Define the interaction-picture state $\tilde\rho(t)\equiv\CU_{\mathrm{lab}}(t)^{-1}[\rho(t)]$. In this picture Eq.~\eqref{eq:axiom} gives
\begin{equation}
  \CU_{\mathrm{lab}}(t)^{-1}\!\circ\CE^{(t)}_M\circ\CU_{\mathrm{lab}}(t)
  =\CE\otimes\id_{\Mbar},
\end{equation}
independent of $t$. Hence the interaction-picture evolution is
$\CE_{M_m}\circ\cdots\circ\CE_{M_1}$. Transforming back at $T$ and comparing
in $A$'s frame gives
$\CV(T)^{-1}\circ\CU_{\mathrm{lab}}(T)=\id_A\otimes\CU_{\mathrm{rel}}$.
\end{proof}

Equation~\eqref{eq:reduction} derives, rather than postulates, the $\id\otimes\Urel$ model: the common mode cancels because each register and its oscillator share a worldline. $A$'s oscillator specifies only the \emph{final comparison frame}; the registers may be reunited, or an equivalent phase reference transferred, after storage. During the hold each local error-correction circuit uses only its local oscillator. A continuously shared oscillator is the distinct common-reference protocol discussed below.

\begin{corollary}[Syndrome blindness]\label{cor:blind}
Let $\CE$ be a projective measurement of any stabilizer $S\in\Sgroup_M$, the local stabilizer group of register $M$, in $M$'s own frame at any time $t$. For outcome $\xi=\pm1$ let $\Pi_\xi=(\Id+\xi S)/2$. Its outcome distribution is $\Tr[\Pi_\xi\tilde\rho(t)]$, independent of $\theta_A$, $\theta_B$ and $\delta\theta$.
\end{corollary}

Local syndromes therefore carry no proper-time information at \emph{any} stage; the difference is relational, available only at the frame conversion in Eq.~\eqref{eq:reduction}.

A Pauli on $M$ likewise commutes to $t=0$, so an error and its correction compose to the identity whatever the times (Appendix~\ref{app:timing}): the decoder must know \emph{which} qubit, never \emph{when}. Any history of local Paulis equals the same Paulis at $t=0$ followed by $\id_A\otimes\Urel$.

\emph{Covariance versus protection.}---Theorem~\ref{thm:reduction} says the rotation survives error correction intact, but nothing about whether the code can \emph{register} it. That question has a sharp answer. Let $\Code$ encode one logical qubit with projector $P$, and denote its logical Pauli operators by $\Xbar$ and $\Zbar$. Call $\Code$ \emph{proper-time covariant} if $\Wop\Code\subseteq\Code$; then $\bar W\equiv P\Wop P$ is logical and $\Urel|_\Code=e^{-i\delta\theta\bar W/2}$ is exactly unitary and preserves the code space. With $w_{0,1}$ the eigenvalues of $\bar W$, define $\mathcal{A}\equiv\tfrac12|w_0-w_1|$ as the \emph{amplification}; up to a global phase, the two logical eigenstates then acquire relative rotation angle $\mathcal{A}\delta\theta$.

\begin{theorem}\label{thm:tradeoff}
Let $\Code$ be proper-time covariant. Then (i) $\mathcal{A}\le n$, with equality iff $\Code=\mathrm{span}\{\ket{0}^{\otimes n},\ket{1}^{\otimes n}\}$ up to a relabelling of the local $Z$ bases, i.e.\ the Greenberger--Horne--Zeilinger (GHZ, or bit-flip repetition) code; (ii) if $\Code$ detects every single-qubit dephasing error, $PZ_iP=\lambda_iP$ for all $i$, then $\mathcal{A}=0$; (iii) hence $\mathcal{A}>0$ only if $PZ_iP\not\propto P$ for some $i$.
\end{theorem}

Proof in Appendix~\ref{app:tradeoff}.

Theorem~\ref{thm:tradeoff} is the Knill--Laflamme condition~\cite{KL97} read backwards, consistent with covariant-code no-go results~\cite{EK09,Faist20,WA20,HNPS} and Preskill's precessing-clock argument~\cite{Preskill00}. It makes the trade-off exact: detecting every single-qubit $Z_i$ eliminates the continuous logical generator. Protection and amplification are controlled by the same $PZ_iP$, and no syndrome schedule changes that. The finite-angle response, however, can survive.

\emph{Response of a non-covariant code.}---Expanding Eq.~\eqref{eq:ulab} factor by factor, with $c_\delta\equiv\cos(\delta\theta/2)$ and $s_\delta\equiv\sin(\delta\theta/2)$,
\begin{equation}
  \Urel=\prod_{i=1}^n\big(c_\delta\,\Id-i\,s_\delta\,Z_i\big)
  =\sum_{K\subseteq[n]}(-i)^{|K|}c_\delta^{\,n-|K|}s_\delta^{\,|K|}Z_K,
  \label{eq:binomial}
\end{equation}
Here $[n]\equiv\{1,\ldots,n\}$ and $Z_K\equiv\prod_{i\in K}Z_i$; the subset $K$ records which qubits contributed their $Z$. This is a binomial expansion and nothing more: the metric rotation is a coherent superposition of $Z$-type errors of every weight. If $Z_K$ anticommutes with a stabilizer $S_j$ then $PZ_KP=PZ_KS_jP=-PZ_KP=0$, so the code space annihilates every term whose $Z_K$ is detectable, acting as a \emph{filter on the weight distribution} of Eq.~\eqref{eq:binomial}.

\begin{theorem}\label{thm:response}
For a stabilizer code with stabilizer group $\Sgroup$ and normaliser $\Ngroup(\Sgroup)$,
\begin{equation}
  P\Urel P=\!\!\sum_{\substack{K\subseteq[n]\\ Z_K\in\Ngroup(\Sgroup)}}\!\!(-i)^{|K|}c_\delta^{\,n-|K|}s_\delta^{\,|K|}\;\overline{Z_K},
  \label{eq:response}
\end{equation}
$\overline{Z_K}\equiv PZ_KP$ the induced logical operator, $=\pm P$ when $Z_K\in\Sgroup$. The leading non-scalar term is at $|K|=d_Z$, the minimum weight of a $Z$-type logical operator, so the logical response is $O(s_\delta^{d_Z})$.
\end{theorem}

Read physically: $c_\delta^nP$ is the amplitude for the metric to have done nothing, stabilizer terms renormalise it, logical terms are the signal, and the rest projects into a nontrivial syndrome sector. A code is covariant exactly when \emph{nothing} is filtered. Figure~\ref{fig:filter} illustrates this filtering for the phase-flip code, for which the only nontrivial $Z$-type normaliser element is $Z_1Z_2Z_3=\Xbar$, so only $|K|=0,3$ survive and
\begin{equation}
  P\Urel P=c_\delta^3P+i\,s_\delta^3\Xbar,
  \label{eq:pf}
\end{equation}
a rotation about the code's logical $\Xbar$ axis, not its $\Zbar$ axis.

\begin{figure}[t]
  \centering
  \includegraphics[width=\columnwidth]{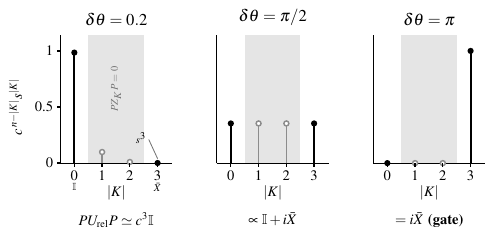}
  \caption{Code-space filtering for the phase-flip $[[3,1,3]]$ code. Detectable weights one and two vanish under $P\Urel P$; only the identity and logical-$\Xbar$ terms survive. At $\delta\theta=\pi$ the excursion is exactly $i\Xbar$.}
  \label{fig:filter}
\end{figure}

\emph{An excursion is a logical gate.}---
\begin{corollary}\label{cor:gate}
Suppose $Z^{\otimes n}\in\Ngroup(\Sgroup)\setminus\Sgroup$. A worldline excursion reaching $\omega_0\Delta\tau=\pi$, by held separation or transport, gives $c_\delta=0$ and $s_\delta=1$, so every term of Eq.~\eqref{eq:response} vanishes except $|K|=n$:
\begin{equation}
  P\Urel P=(-i)^n\,\overline{Z^{\otimes n}}.
\end{equation}
This is a \emph{deterministic logical gate}: unit code-space survival and trivial syndrome. Being a logical Pauli it acts trivially on its own eigenstates, so the readout state must not be one; the logical Bell state $\ket{0_L1_L}+\ket{1_L0_L}$ maps to $\ket{0_L0_L}+\ket{1_L1_L}$, which $\Zbar_A\Zbar_B$ separates at full contrast.
\end{corollary}

Among the examples in Appendix~\ref{app:gatecond}, the hypothesis holds for the phase-flip, bit-flip, Steane, Shor, $[[5,1,3]]$ and odd-distance rotated surface codes, while $[[4,2,2]]$ and codes containing a stabilizer with odd $X/Y$ parity provide distinct failures; which logical operator results is code-dependent.

Nor is it fine-tuned, and the way it fails is the point. At $\delta\theta=\pi/2^k+\epsilon$, with integer $k\ge0$, the terms of Eq.~\eqref{eq:response} no longer conspire, and the deviation splits into two channels of very different order. The logical-state-averaged probability of remaining in the code space is
\begin{equation}
  \Tr[(P\Urel P)^{\dagger}(P\Urel P)]/2 = 1-\kappa\,\epsilon^{2}+O(\epsilon^{4}),
  \label{eq:survival}
\end{equation}
so a fraction $\kappa\epsilon^2+O(\epsilon^4)$ is projected into nontrivial syndrome sectors. These events are detected by syndrome extraction and may be discarded; $\kappa=3/4$ for the phase-flip code at $\pi$ and $15/4$ for $[[15,1,3]]$ at $\pi/4$. Appendix~\ref{app:survival} derives Eq.~\eqref{eq:survival} and the coefficients. Conditioned on the accepted trivial-syndrome branch, the entanglement (process) fidelity is far higher: the induced operator departs from the ideal gate only at order $\epsilon^{2d_Z}$, for the phase-flip code
\begin{equation}
  F_{\Xbar}^{\rm cond}=\frac{s_\delta^6}{c_\delta^6+s_\delta^6}=1-(\epsilon/2)^6+O(\epsilon^{8}),
  \label{eq:robust}
\end{equation}
and for $[[15,1,3]]$ the conditional infidelity likewise scales as $\epsilon^{6}$; both statements are derived in Appendix~\ref{app:survival}. For these distance-three examples, miscalibration therefore primarily reduces the accepted-shot probability at second order, while the conditional logical error of accepted shots is parametrically higher order. This is the $O(s_\delta^{d_Z})$ suppression of Theorem~\ref{thm:response} at the other end of the range: a poor analogue phase meter is for that reason an excellent digital gate.

The construction ladders. At $\delta\theta=\pi/2^k$, a Calderbank--Shor--Steane (CSS) code whose full binary $X$-stabilizer support space is $2^{k+1}$-divisible, with each logical coset of fixed Hamming weight modulo $2^{k+1}$, admits the corresponding transversal $Z$ rotation. This is a property of the full support spaces, not merely of the weights of a chosen generator set. Thus $[[15,1,3]]$ at $\pi/4$ returns an exact logical $\Tbar^\dagger$ in the convention used in this Letter even though $P\Wop P=0$, so its protected code-space response vanishes at first order; Appendix~\ref{app:gatecond} gives the finite-angle derivation.

\begin{corollary}[Clock-off and clock-on]\label{cor:switch}
Call an encoded pure state \emph{continuously clock-off} when proper-time evolution changes it only by an overall phase,
$\Urel(\delta\theta)\ket{\mathrm{off}}=e^{i\phi_{\rm off}(\delta\theta)}\ket{\mathrm{off}}$ for all $\delta\theta$. Equivalently, $\ket{\mathrm{off}}$ is a $\Wop$ eigenstate. A locally observable clock requires coherent support on sectors with distinct $\Wop$ eigenvalues; for a proper-time-covariant logical qubit their relative phase rate is $\mathcal{A}=\tfrac12|w_0-w_1|$. More generally, let a non-covariant code admit an exact finite excursion $\delta\theta_0$ with $\Urel(\delta\theta_0)P=U_LP$, where $U_L$ is the induced logical unitary. A state is clock-off at that excursion if $U_L\ket{\mathrm{off}}=e^{i\phi_{\rm off}}\ket{\mathrm{off}}$; an entire subspace is clock-off when the same eigenphase holds throughout it. Clock-on states are coherent superpositions of logical eigensectors with distinct eigenphases.
\end{corollary}

Corollary~\ref{cor:switch} also separates an ordinary clock comparison from a genuinely interferometric application. A single encoded register read against a remote local oscillator is operationally an error-corrected clock comparison: the theorem still applies, but the relative phase may be attributed to the two classical reference clocks. The stronger use is an encoded analogue of the clock/no-clock state of Ref.~\cite{Borregaard2025},
\begin{equation}
  \ket{\Psi_{\rm cn}}\propto
  \ket{\mathrm{off}}_A\ket{\mathrm{on}}_B+
  \ket{\mathrm{on}}_A\ket{\mathrm{off}}_B,
\end{equation}
for which the identity of the ticking memory is itself coherently superposed and proper time changes a branch-interference observable rather than merely the phase of an analyser.

\emph{Numbers.}---Take two registers at fixed heights differing by $\Delta h$. For $^{137}$Ba$^+$ hyperfine qubits, $\omega_0=2\pi\times8.037$~GHz,
$$
  \frac{d\,\delta\theta}{dt}=\frac{\omega_0\,g\,\Delta h}{c^{2}},
  \qquad
  T_\pi=\frac{\pi c^{2}}{\omega_0\,g\,\Delta h}=\frac{5.70\times10^{5}\ \mathrm{m\,s}}{\Delta h}.
$$
The rate is $5.51\times10^{-6}$~rad\,s$^{-1}$\,m$^{-1}$. For scale, the Ba$^+$ transition is a concrete benchmark rather than a requirement: $T_\star\propto\omega_0^{-1}$. An encoded two-level optical transition at $193$~THz would reach $\delta\theta=\pi/2$ at $\Delta h=10$~m in about $1.2$~s, versus $7.9$~h for the $8.037$~GHz transition. One parameter follows, whose two ends behave oppositely: at small $\delta\theta$ the phase is a syndrome-invisible perturbation, while at admitted finite excursions it closes exactly into a logical gate.

\emph{Small $\delta\theta$: an error no syndrome reports.} For a pure encoded probe $\ket{\Psi}$, the state overlap after the relational rotation on $B$ obeys
$1-\big|\bra{\Psi}(\Id_A\otimes\Urel)\ket{\Psi}\big|^2
=\frac{\mathcal{F}_Q}{4}\,\delta\theta^2+O(\delta\theta^4)$,
where $\mathcal{F}_Q=4\,\mathrm{Var}_{\Psi}(\Wop_B/2)$ is the pure-state quantum Fisher information; the expansion is derived in Appendix~\ref{app:survival}.
Thus the size of the logical disturbance is code- and state-dependent even though its syndrome blindness is universal (Corollary~\ref{cor:blind}). No local correction reveals or removes the relational phase. Whether this matters turns on the reference. A master oscillator at height $h_0$ distributed to both gives node $M$ a detuning $\Delta_M\equiv\omega_0g(h_0-h_M)/c^2$, leaving the differential rate unchanged but making it a genuine collective error. Over one syndrome period $\tau_{\mathrm{QEC}}$ the corresponding phase is $\phi_M=\Delta_M\tau_{\mathrm{QEC}}$, which syndrome extraction filters to $O(\phi_M^{d_Z})$. A processor needing such a reference anyway for gates between modules is thereby protected, by the same mechanism that freezes the effect in Ref.~\cite{Preskill00}.

For independent network nodes the relational phase remains absent from every local syndrome record. Its observable effect is instead set by the encoded probe: at small $\delta\theta$ the state change is governed by $\mathcal{F}_Q$ as above. Thus the syndrome record supplies neither an estimate of $\Delta\tau$ nor a trigger for correcting it; the accumulated relation is revealed only when the two frames are finally compared.

\emph{Finite excursion: a gate the code admits.} At the $450$~m Tokyo Skytree scale~\cite{Takamoto20}, the target holds are $1267$, $634$, and $317$~s for $\delta\theta=\pi,\pi/2,\pi/4$. These times can be traded against encoding resources: with a suitable logical-basis choice, concatenating each outer-code qubit with a $b$-qubit GHZ block amplifies the physical phase to $b\,\delta\theta$ at the outer-code level, so a gate admitted by the outer code at angle $\theta_0$ is reached at $\delta\theta=\theta_0/b$. The required hold is therefore shorter by a factor $b$, at the corresponding resource and noise-sensitivity cost (Appendix~\ref{app:family}). Dynamical decoupling cannot protect a quasi-static phase while retaining the same quasi-static signal. For a hold of duration $T$ with syndrome period $\tau_{\mathrm{QEC}}$, define $N\equiv T/\tau_{\mathrm{QEC}}$ rounds; the relevant requirement is then $Np_L\ll1$, with $p_L$ the logical failure probability per round.

This Letter makes no stronger circuit-level claim: $p_L$ depends on the extraction gadget, ancilla preparation and verification, leakage, correlated faults, schedule, and decoder. Appendix~\ref{app:gatecond} therefore lists examples, not an optimal code set. Any code and extraction architecture satisfying the algebraic gate condition and the required logical lifetime could serve; other constructions may well be experimentally superior. Collective phase offsets act through the same generator $\Wop$ and must be tracked rather than corrected (Appendix~\ref{app:collective}).

\emph{Discussion.}---Proper time acts on a pair of locally corrected memories as a relational $U(1)$ that error correction cannot see and cannot remove, exactly, for any code and any history of local frame-covariant operations/errors satisfying Axiom~\ref{ax:comoving}. Corollary~\ref{cor:gate} is what that buys: on any code satisfying its normaliser condition, a worldline excursion is a clean logical operation the code neither suppresses nor notices. Theorem~\ref{thm:tradeoff} is the price of that blindness: noise and signal share the generator $\Wop$, so a code protecting against dephasing cannot register the rotation at first order, the stabilizer form of the obstruction that denies Heisenberg scaling to a signal parallel to its noise~\cite{ZZPJ18}. The new ingredients of this Letter are the exact distinct-worldline factorisation of Theorem~\ref{thm:reduction} and the conversion of the surviving finite proper-time phase into a discrete logical gate (Appendix~\ref{app:gatecond}). Relativistic motion has independently been proposed as a way to parameterise physical quantum gates~\cite{LeMaitre25}; here the proper-time phase survives local QEC as a relational logical operation.

Preskill's conclusions separate naturally into common-frame and code-structure statements. His first two hold in a common rest frame, where Theorems~\ref{thm:reduction} and~\ref{thm:tradeoff} recover them, but do not carry over unchanged to distinct worldlines: nothing is frozen, and the surviving phase is the physical relation $\omega_0\Delta\tau$, not a convention. His remaining conclusions, concerning covariance and the protected channel, are recovered by Theorem~\ref{thm:tradeoff} and the finite-angle response developed here.

Operationally, Eq.~\eqref{eq:reduction} requires no common oscillator during the hold; the registers are placed in one comparison frame only at final readout, by reunion or an equivalent final phase-reference transfer. If reunion is used, there is at least a trapped-particle transport precedent: a cryogenic Penning trap retaining $10^2$ protons has been transported $3.7$~km by road~\cite{BASESTEP25}. At the discrete gate points the transformed logical Bell states can be distinguished by logical-$Z$ correlations in the final comparison frame, whereas generic continuous phase readout requires an equatorial reference there. The phase is therefore not only a rate to be estimated: at admitted values $\pi/2^k$, the same suppression that defeats first-order metrology closes into a deterministic gate. Within this sequence, smaller admitted phases climb the Clifford hierarchy, reaching a $T$-type gate at $\pi/4$.

Theorem~\ref{thm:reduction} also supplies a sharp, model-independent null test. Within its assumptions, arbitrary local frame-covariant operations and errors are absorbed into the bare local circuit, leaving only $\Id_A\otimes\Urel$ as the relational residue. Because this residue is local and unitary, the proper-time contribution can change the relational phase but cannot change the $A|B$ entanglement of an initially encoded logical Bell state, for any hold time or accumulated $\delta\theta$. Corollary~\ref{cor:blind} independently implies that the proper-time contribution cannot alter the local syndrome statistics. These give two complementary diagnostics of dynamics outside the reduced channel: entanglement loss beyond that predicted by the independently characterised logical-memory noise channel, or a gravitationally correlated excess or change in nontrivial syndrome events. Either would signal evolution not representable as $\Id_A\otimes\Urel$. Stochastic or non-unitary gravitational modifications, including collapse-like or metric-diffusion-like effective channels, provide examples of departures that could violate this null prediction, but their microscopic form and magnitude are model dependent. Conventional objective-collapse models, including Di\'osi--Penrose and CSL-type theories, are normally formulated for superpositions of spatial mass density, so the present internal-clock-memory configuration is not automatically their natural regime~\cite{Diosi89,Penrose96,Bassi13}.

Two extensions follow naturally from the present qubit result. First, the clock-off/clock-on sectors above provide the ingredients for a protected analogue of the Borregaard--Pikovski clock/no-clock interferometer. In Ref.~\cite{Borregaard2025} the no-clock branch is stationary, up to an overall phase, throughout the proper-time evolution, whereas the non-covariant constructions considered here guarantee that invariance only at admitted finite excursions; constructing the full protected interference protocol and its visibility observable is therefore a separate task. Second, the scaling $T_\star\propto\omega_0^{-1}$ makes optical-frequency memories especially attractive. The theorem as stated, however, applies to qubit registers with locally referenced two-level transitions; extending the frame-reduction argument to bosonic oscillator encodings or measurement-based architectures requires a corresponding covariance analysis.

A fault-tolerant memory is built to hold a state still, and between worldlines it cannot: it keeps proper time exactly, and no syndrome records that it does. This turns error-corrected memories into a platform for experiments at the interface of quantum coherence and general-relativistic proper time. A single-register readout is an encoded clock comparison; entangled logical memories add nonlocal null tests, while clock/no-clock extensions suggest making the clock assignment itself quantum and converting the same relational unitary into proper-time interference. Discrete excursions additionally supply logical gates. The coherence that kept these regimes out of reach is supplied by a code rather than isolation. Left to run, the phase is an error a network cannot see; read on a code that admits it, proper time itself supplies the logical gate.

\begin{acknowledgments}
% SPK: confirm each name below is happy to be acknowledged before submission.
We thank O.~Erkilic, J.~Foo and D.~Ramachandran for comments on the manuscript,
L.~O.~Conlon for discussions on the comparison with unencoded interferometry,
and T.~C.~Ralph, P.~S.~Turner, G.~K.~Brennen and A.~D.~Greentree for discussions.
OpenAI ChatGPT (GPT-5.6 Sol) was used to assist with manuscript editing, literature cross-checking, and technical consistency checks. Anthropic Claude (Opus 5.5) was used to independently assess scientific arguments and assist with manuscript revision. In all cases, the authors directed the queries and independently verified the scientific content.
\end{acknowledgments}

\paragraph*{Data availability.} All numerical values reported in this Letter and its appendices follow directly from the equations presented here; no research data sets were created or analyzed.

\appendix
\setcounter{secnumdepth}{2}
\renewcommand{\thesection}{\Alph{section}}

% ---------------- END MATTER (counts against the 2-page allowance) ----------------
\section{Frame covariance of local noise}\label{app:rwa}

To illustrate the extension beyond control operations, consider local noise represented semiclassically by time-dependent fields. In $M$'s rest frame let $\Omega_i^{x,y,z}(t)$ denote their components, so the single-qubit Hamiltonian is $H_i=\tfrac{\hbar\omega_0}{2}Z_i+\tfrac{\hbar}{2}[\Omega^x_iX_i+\Omega^y_iY_i+\Omega^z_iZ_i]$. Longitudinal noise commutes with $\Wop$. For the transverse part, pass to the frame rotating at $\omega_0$: $R_Z(\omega_0\tau)X_iR_Z(-\omega_0\tau)=\cos(\omega_0\tau)X_i+\sin(\omega_0\tau)Y_i$, so off-resonant components average away. Writing a near-resonant component with amplitude $\Omega_{i,0}$ as $\Omega^x_i=\Omega_{i,0}\cos(\omega_0\tau_M+\psi_i)$ and discarding the counter-rotating term gives a stationary equatorial generator. Undoing the frame transformation at $t_e$ yields
\begin{equation}
  E_i(t_e)=R_Z(\beta+\psi_i)\,X_i\,R_Z(-\beta-\psi_i),
  \qquad \beta=\theta_M(t_e),
  \label{eq:noiseaxis}
\end{equation}
which has the covariance form of Axiom~\ref{ax:comoving} for fixed event offset $\psi_i$.

The redshift enters through the resonance condition: a field at coordinate frequency $\omega_{\rm coord}$ is seen on $\gamma_M$ at $\omega_{\rm coord}\,dt/d\tau_M$, so only spectral weight near the qubit's lab frequency $\omega_0\,d\tau_M/dt$ drives it. Resonance pins $\omega_{\rm coord}$ only to within a linewidth, so $\psi_i$ is constant over a single flip but is not pinned to any particular worldline's frame across events; Eq.~\eqref{eq:noiseaxis} asserts the equatorial form and nothing stronger.

\section{When the excursion is a gate}\label{app:gatecond}

Corollary~\ref{cor:gate} needs $Z^{\otimes n}\in\Ngroup(\Sgroup)\setminus\Sgroup$. For a general Pauli stabilizer, $Z^{\otimes n}$ commutes with it iff the number of $X$ or $Y$ factors is even; for a CSS code this reduces to even weight of every $X$-type stabilizer. The second condition excludes the degenerate case in which $Z^{\otimes n}$ is itself a stabilizer and the gate is the identity. For the smaller-angle CSS rows below, the ``divisor'' is a property of the full binary space generated by the $X$-type stabilizers, together with the fixed weight residues of the logical cosets; the listed check weights are only representative generators.

The smaller-angle condition follows directly in the computational basis. With $\mathrm{wt}(\mathbf{x})$ the Hamming weight of bit string $\mathbf{x}$,
$R_Z(\theta)^{\otimes n}\ket{\mathbf{x}}=e^{-in\theta/2}e^{i\theta\,\mathrm{wt}(\mathbf{x})}\ket{\mathbf{x}}$. At $\theta=\pi/2^k$ with integer $k\ge0$, a CSS logical branch therefore acquires a common phase iff every word in its $X$-stabilizer support differs in weight by a multiple of $2^{k+1}$ and each logical coset has a fixed weight residue modulo $2^{k+1}$. Under precisely this condition the excursion preserves the code space and induces a logical $Z$ rotation with unit survival. For the $[[15,1,3]]$ punctured Reed--Muller code, the $X$-stabilizer support is the $8$-divisible $[15,4,8]$ simplex code: $\ket{0_L}$ has weights $0,8$ and $\ket{1_L}$ weights $7,15$, so at $\theta=\pi/4$ the two branches differ by $e^{-i\pi/4}$, giving $\Tbar^\dagger$ in the convention of this Letter. To fix the gate signs in the table, for the one-logical-qubit CSS rows we choose $\ket{0_L}$ as the divisible $X$-stabilizer superposition and $\ket{1_L}=X^{\otimes n}\ket{0_L}$; the logical-one coset then has weight residue $n$ modulo the listed divisor.

Here $d$ denotes the code distance.

\begin{widetext}
\begin{center}
\begin{tabular}{lcccccll}
\hline\hline
code & $n$ & $d$ & representative $X$-gen. weights & space divisor & $\delta\theta$ & gate & hold (s) \\
\hline
\multicolumn{8}{l}{\emph{doubly even (4-divisible) examples}}\\
Steane $[[7,1,3]]$ & $7$ & $3$ & $4,4,4$ & $4$ & $\pi/2$ & $\Sbar^\dagger$ & $634$ \\
$[[17,1,5]]$ & $17$ & $5$ & $4,8$ & $4$ & $\pi/2$ & $\Sbar$ & $634$ \\
Golay $[[23,1,7]]$ & $23$ & $7$ & $8,12,16$ & $4$ & $\pi/2$ & $\Sbar^\dagger$ & $634$ \\
$[[47,1,11]]$ & $47$ & $11$ & $\equiv0\ (4)$ & $4$ & $\pi/2$ & $\Sbar^\dagger$ & $634$ \\
\hline
\multicolumn{8}{l}{\emph{triply even (8-divisible)}}\\
Reed--Muller $[[15,1,3]]$ & $15$ & $3$ & $8^{\times4}$ & $8$ & $\pi/4$ & $\Tbar^\dagger$ & $317$ \\
Reed--Muller $[[31,1,3]]$ & $31$ & $3$ & $16^{\times5}$ & $16$ & $\pi/8$ & $\Tbar^{-1/2}$ & $158$ \\
\hline
\multicolumn{8}{l}{\emph{Pauli only, or degenerate}}\\
phase-flip $[[3,1,3]]$ & $3$ & $3$ & $2,2$ & $2$ & $\pi$ & $\Xbar$ & $1267$ \\
Shor $[[9,1,3]]$ & $9$ & $3$ & $6,6$ & $2$ & $\pi$ & $\Xbar$ & $1267$ \\
$[[5,1,3]]$ & $5$ & $3$ & (non-CSS) & --- & $\pi$ & $\Zbar$ & $1267$ \\
bit-flip $[[3,1,3]]$ & $3$ & $3$ & --- & --- & any (covariant) & $R_Z(n\delta\theta)$ & --- \\
$[[4,2,2]]$ & $4$ & $2$ & $4$ & --- & --- & identity & --- \\
odd $X/Y$ parity & --- & --- & $3$ & --- & $\pi$ & not preserved & --- \\
\hline\hline
\end{tabular}
\end{center}
\end{widetext}

The familiar distance-three examples above clear the condition at $\delta\theta=\pi$. How far below $\pi$ a CSS code reaches is set by the divisibility of its \emph{full} $X$-stabilizer support space and the weight residues of its logical cosets. If that support space is $2^{k+1}$-divisible and each logical coset has fixed weight modulo $2^{k+1}$, then $\delta\theta=\pi/2^k$ preserves the code space and induces the corresponding logical $Z$ rotation; generator rows of weight $2^{k+1}$ are not by themselves sufficient. A $4$-divisible support space gives the Clifford $S$-type working point, while the third level first appears for the appropriate $8$-divisible structure. Codes below their divisibility limit do not merely lose contrast, they leave the code space: the phase-flip code retains only $1/4$ of its weight at $\pi/2$, and $[[5,1,3]]$ only $1/16$. The two failures at $\pi$ are distinct: in $[[4,2,2]]$, $Z^{\otimes4}$ is a stabilizer and the metric acts trivially; if a stabilizer has odd $X/Y$ parity, $Z^{\otimes n}$ does not preserve the code space at the excursion.

\section{Proof of Theorem~\ref{thm:tradeoff}}\label{app:tradeoff}

(i) $\|\Wop\|=n$ so $|w_{0,1}|\le n$; equality forces $\{w_0,w_1\}=\{n,-n\}$, whose eigenspaces are the one-dimensional spans of $\ket{0}^{\otimes n}$ and $\ket{1}^{\otimes n}$. (ii) $PZ_iP=\lambda_iP\;\forall i$ gives $P\Wop P=(\sum_i\lambda_i)P$, so $\bar W$ is proportional to the logical identity and $w_0=w_1$. (iii) is the contrapositive of (ii). The converse fails: on $\{\ket{01},\ket{10}\}$, $PZ_1P=\Zbar$ and $PZ_2P=-\Zbar$ are both non-scalar, yet they sum to zero and $\mathcal{A}=0$.

% ---------------------------------------------------------------------------
% ARXIV APPENDICES.  Appendices D--H incorporate the additional derivations
% and technical analyses used by the main text.
% ---------------------------------------------------------------------------

\section{The offsets $\psi_i$ and their erasure}\label{app:offsets}

For comparison, Appendix~\ref{app:timing} gives the timing-independence identity for ideal Pauli corrections. For a physical equatorial flip write $E_i=R_Z(\beta+\psi_i)X_iR_Z(-\beta-\psi_i)$, with $\beta$ the local oscillator phase at the event; the additional event offset is $\psi_i$. A sequence of $r$ flips on one qubit with no intervening syndrome extraction leaves
\begin{equation}
  R_Z\Big(\theta_M+2\textstyle\sum_{j=1}^{r}(-1)^{r-j}\psi_j\Big)X^r:
\end{equation}
the times cancel identically (by Theorem~\ref{thm:reduction}), the offsets random-walk.

Let $P$ denote the code projector and $\mathcal{A}$ the amplification defined in the main text. Syndrome extraction can discretise the offsets, and the criterion for when is exact. Write the event as $\cos\psi\,X_i+\sin\psi\,Y_i$. The measurement collapses it onto the $X_i$ or $Y_i$ branch, and correcting the flagged branch leaves a pure Pauli with $\psi$ destroyed, if and only if the two branches are distinguishable, i.e.\ $X_i$ and $Y_i$ have distinct syndromes. Since $Y_i=iX_iZ_i$, this holds iff $Z_i$ has a nontrivial syndrome, i.e.\ iff $PZ_iP=0$. There are then three cases, not two. If $PZ_iP=0$ the offset is erased and $Z_i$ is detectable; if $PZ_iP=\lambda_iP$ with $\lambda_i\neq0$ the offset survives but qubit $i$ is blind to the metric, contributing nothing to $\mathcal{A}$; and only if $PZ_iP$ is non-scalar does the offset survive \emph{and} the qubit register, as a residual logical $R_Z(2\psi_i)$. Erasure therefore implies blindness, but blindness does not imply erasure.
Theorem~\ref{thm:tradeoff} is thus simultaneously the covariance--protection trade-off and the $\psi$-erasure criterion. In the phase-flip code $X_i\mapsto(+1,+1)$ and $Y_i\mapsto(-1,+1)$: every $\psi_i$ is erased, and the code is correspondingly signal-blind. In the bit-flip code $Z_i$ is logical, $X_i$ and $Y_i$ are syndrome-degenerate, and $\psi_i$ survives on every qubit, which is the unavoidable cost of $\mathcal{A}=n$. The practical mitigation is not a code but a rate: for hyperfine ground-state qubits the flip channel has no spontaneous component, so the flip probability over the hold is $p_X\simeq T_\pi/T_1$, where $T_\pi$ is the hold duration required to accumulate a relative phase of $\pi$ and $T_1$ is the population-relaxation time.

\section{Timing independence of a correction}\label{app:timing}

With $\theta_M(t)$ the phase accumulated by register $M$ relative to its local oscillator, $R_Z(\alpha)\equiv e^{-i\alpha Z/2}$ and $X_\beta\equiv R_Z(\beta)XR_Z(-\beta)$, $\beta_{e,c}=\theta_M(t_{e,c})$, and $\theta\equiv\theta_M(T)$ the final accumulated phase,
\begin{equation}
  R_Z(\theta{-}\beta_c)X_{\beta_c}R_Z(\beta_c{-}\beta_e)X_{\beta_e}R_Z(\beta_e)=R_Z(\theta),
  \label{eq:echoid}
\end{equation}
independently of $t_e$ and $t_c$. Two reflections about equatorial axes separated by $\beta_c-\beta_e$ compose into a rotation by $2(\beta_c-\beta_e)$, exactly the phase the qubit failed to accumulate while inverted. The decoder must therefore know \emph{which} qubit, never \emph{when}. It is the correction pulse's phase reference, not merely its inversion of the populations, that restores the signal: a correction about an axis fixed in coordinate time does not satisfy Eq.~\eqref{eq:echoid}.

\section{Derivations of survival and conditional gate fidelity}\label{app:survival}

This section supplies the intermediate steps behind the finite-angle robustness formulas quoted in the main text. Let
\begin{equation}
  U(\theta)\equiv e^{-i\theta\Wop/2},
  \qquad Q\equiv \Id-P,
\end{equation}
with $P$ the projector onto a one-logical-qubit code, so the logical-space dimension is $d_L\equiv\Tr P=2$. Suppose $\theta_0$ is an exact working point at which $U(\theta_0)$ preserves the code and induces a logical unitary $U_L$, i.e. $U(\theta_0)P=U_LP$. Because all $U(\theta)$ are generated by the same $\Wop$,
\begin{equation}
  PU(\theta_0+\epsilon)P
  =U_L\,\Gamma(\epsilon),
  \qquad
  \Gamma(\epsilon)\equiv Pe^{-i\epsilon\Wop/2}P .
  \label{eq:s5factor}
\end{equation}
Thus the ideal gate can be factored off and all miscalibration properties are determined by $\Gamma(\epsilon)$.

\paragraph{Code-space survival.}---For the maximally mixed logical input $P/d_L$, the survival probability is
\begin{equation}
  p_{\rm surv}(\epsilon)
  =\frac{1}{d_L}\operatorname{Tr}\!\left[\Gamma(\epsilon)^\dagger\Gamma(\epsilon)\right].
  \label{eq:s5survdef}
\end{equation}
Unitarity gives the complementary probability of leaving the code space as
\begin{equation}
  1-p_{\rm surv}
  =\frac{1}{d_L}\operatorname{Tr}\!\left[
    P e^{i\epsilon\Wop/2}Qe^{-i\epsilon\Wop/2}P
  \right].
\end{equation}
Since
\begin{equation}
  Qe^{-i\epsilon\Wop/2}P
  =-\frac{i\epsilon}{2}Q\Wop P+O(\epsilon^2),
\end{equation}
the leading reduction in code-space survival is
\begin{equation}
  p_{\rm surv}(\epsilon)
  =1-\kappa\epsilon^2+O(\epsilon^4),
  \qquad
  \kappa=\frac{1}{4d_L}\operatorname{Tr}\!\left[P\Wop Q\Wop P\right].
  \label{eq:s5kappa}
\end{equation}
The absence of odd powers follows also from $p_{\rm surv}(\epsilon)=p_{\rm surv}(-\epsilon)$. If every weight-one and weight-two $Z$ operator is detectable, then $P\Wop P=0$ and
$P\Wop^2P=nP$, because
$\Wop^2=n\Id+2\sum_{i<j}Z_iZ_j$. Equation~\eqref{eq:s5kappa} then reduces to
\begin{equation}
  \kappa=\frac{n}{4}.
\end{equation}
This gives $\kappa=3/4$ for the phase-flip $[[3,1,3]]$ code at $\theta_0=\pi$ and $\kappa=15/4$ for the $[[15,1,3]]$ Reed--Muller code at $\theta_0=\pi/4$.

\paragraph{Phase-flip code.}---For the phase-flip code the projected rotation is exactly
\begin{equation}
  \Gamma_{\rm pf}(\theta)=c_\theta^3\Id_L+i s_\theta^3\Xbar,
  \qquad
  c_\theta\equiv\cos(\theta/2),\quad s_\theta\equiv\sin(\theta/2),
  \label{eq:s5pfK}
\end{equation}
where $\Id_L$ is the logical identity. Since the two logical Pauli components are Hilbert--Schmidt orthogonal,
\begin{equation}
  \Gamma_{\rm pf}^\dagger\Gamma_{\rm pf}=(c_\theta^6+s_\theta^6)\Id_L,
  \qquad
  p_{\rm surv}=c_\theta^6+s_\theta^6.
  \label{eq:s5pfsurv}
\end{equation}
At $\theta=\pi+\epsilon$, $c_\theta=-\sin(\epsilon/2)$ and $s_\theta=\cos(\epsilon/2)$, giving
\begin{equation}
  p_{\rm surv}
  =1-\frac{3}{4}\epsilon^2+\frac14\epsilon^4+O(\epsilon^6),
\end{equation}
which is the $\kappa=3/4$ case of Eq.~\eqref{eq:s5kappa}.

For a postselected one-Kraus logical map $\rho\mapsto\Gamma\rho\Gamma^\dagger$, we quantify the conditional gate accuracy by the entanglement (process) fidelity with an ideal target logical unitary $U_{\rm tar}$,
\begin{equation}
  F_{\rm cond}(U_{\rm tar},\Gamma)
  \equiv
  \frac{\left|\operatorname{Tr}_L(U_{\rm tar}^\dagger\Gamma)\right|^2}
       {d_L\,\operatorname{Tr}_L(\Gamma^\dagger\Gamma)}.
  \label{eq:s5process}
\end{equation}
Here $\operatorname{Tr}_L$ denotes the trace over the logical code space. At $\theta_0=\pi$ the ideal target is $U_{\rm tar}=i\Xbar$ up to global phase. Equations~\eqref{eq:s5pfK} and \eqref{eq:s5process} give
\begin{equation}
\begin{aligned}
  F_{\Xbar}^{\rm cond}
  &=\frac{s_\theta^6}{c_\theta^6+s_\theta^6}\\
  &=1-\frac{\epsilon^6}{64}-\frac{\epsilon^8}{128}
    +O(\epsilon^{10})\\
  &=1-(\epsilon/2)^6+O(\epsilon^8).
\end{aligned}
  \label{eq:s5FX}
\end{equation}
This is the expression quoted in the main text. The quadratic reduction in survival and sixth-order conditional infidelity are therefore distinct quantities: most small miscalibrations produce a nontrivial syndrome and are rejected, while the accepted trivial-syndrome branch remains much closer to the ideal gate.

The same projected operator also explains the $\theta_0=\pi/2$ entry used in Appendix~\ref{app:collective}. At that point the normalized postselected gate is
$U_{\rm tar}^{(\sqrt{\Xbar})}=(\Id_L+i\Xbar)/\sqrt2$ and the success probability is only $p_{\rm surv}=1/4$. For $\theta=\pi/2+\epsilon$,
\begin{equation}
\begin{aligned}
  F_{\sqrt{\Xbar}}^{\rm cond}
  &=\frac{(c_\theta^3+s_\theta^3)^2}{2(c_\theta^6+s_\theta^6)}
    =1-\frac94\epsilon^2+O(\epsilon^4),\\
  p_{\rm surv}
  &=\frac14+\frac34\epsilon^2+O(\epsilon^4).
\end{aligned}
  \label{eq:s5FsqrtX}
\end{equation}
Thus $\pi/2$ is a postselected $\sqrt{\Xbar}$ point for this code, not a deterministic gate point. As another exact check used in Appendix~\ref{app:gatecond}, the $[[5,1,3]]$ code has only the identity and $Z^{\otimes5}$ among its $Z$-only normaliser elements; hence at $\theta=\pi/2$,
$p_{\rm surv}=c_\theta^{10}+s_\theta^{10}=1/16$.

\paragraph{Why the conditional infidelity scales as $\epsilon^{2d_Z}$.}---After factoring the exact logical gate as in Eq.~\eqref{eq:s5factor}, Theorem~\ref{thm:response} implies
\begin{equation}
\begin{aligned}
  \Gamma(\epsilon)&=a(\epsilon)\Id_L+\Gamma_{\rm ns}(\epsilon),\\
  \operatorname{Tr}_L \Gamma_{\rm ns}(\epsilon)&=0,
  \qquad \Gamma_{\rm ns}(\epsilon)=O(\epsilon^{d_Z}).
\end{aligned}
\end{equation}
where $a(\epsilon)$ collects the identity and stabilizer contributions and $\Gamma_{\rm ns}$ the non-scalar logical part. Taking $U_{\rm tar}=\Id_L$ in Eq.~\eqref{eq:s5process},
\begin{equation}
  1-F_{\rm cond}
  =\frac{\operatorname{Tr}_L(\Gamma_{\rm ns}^\dagger\Gamma_{\rm ns})}
  {d_L|a|^2+\operatorname{Tr}_L(\Gamma_{\rm ns}^\dagger\Gamma_{\rm ns})}
  =O(\epsilon^{2d_Z}).
  \label{eq:s5dz}
\end{equation}
For the $[[15,1,3]]$ Reed--Muller code, $d_Z=3$, so the conditional infidelity about the exact $\pi/4$ transversal gate begins at $O(\epsilon^6)$ even though the code-space survival falls quadratically according to Eq.~\eqref{eq:s5kappa}.

\paragraph{Small-angle overlap and quantum Fisher information.}---For completeness, the small-angle formula in the main text follows directly from the characteristic function of the generator. For a pure probe state $\ket{\Psi}$, let $G_{\rm sig}\equiv\Wop_B/2$ and $f_\Psi(\delta\theta)\equiv\langle\Psi|e^{-i\delta\theta G_{\rm sig}}|\Psi\rangle$. Expanding,
\begin{equation}
  f_\Psi(\delta\theta)
  =1-i\delta\theta\langle G_{\rm sig}\rangle
   -\frac{\delta\theta^2}{2}\langle G_{\rm sig}^2\rangle
   +O(\delta\theta^3),
\end{equation}
so
\begin{equation}
\begin{aligned}
  1-|f_\Psi(\delta\theta)|^2
  &=\delta\theta^2
    \big(\langle G_{\rm sig}^2\rangle-\langle G_{\rm sig}\rangle^2\big)
    +O(\delta\theta^4)\\
  &=\frac{\mathcal{F}_Q}{4}\delta\theta^2
    +O(\delta\theta^4).
\end{aligned}
\end{equation}
with $\mathcal{F}_Q=4\,\mathrm{Var}_{\Psi}(G_{\rm sig})$ for a pure probe. This also makes explicit why the small-angle state disturbance is probe-dependent even though syndrome blindness is not.

\section{The one-parameter family and its Fisher information}\label{app:family}

Concatenating an inner GHZ code on $b$ qubits inside an outer phase-flip code on $N_{\rm blk}\equiv n/b$ logical qubits gives, on $n$ physical qubits, $\mathcal{A}_{\mathrm{inner}}=b$, $d_X=b$, $d_Z=n/b$, where $d_X$ and $d_Z$ denote the minimum weights of $X$- and $Z$-type logical operators. By part (i) of Theorem~\ref{thm:tradeoff}, the inner GHZ layer is the unique $b$-qubit covariant encoding that achieves the maximal amplification $\mathcal{A}=b$; the outer layer supplies whatever protection is wanted. Moving along this sensitivity--protection interpolation is a change of \emph{which stabilizers are measured}, not a re-encoding: at $b=1$ it is the transversal Hadamard exchanging $\langle X_iX_j\rangle$ for $\langle Z_iZ_j\rangle$; for $b>1$ it is suspension of the outer $X$-type checks.

The same amplification can shorten a finite-angle gate excursion. Within each GHZ block the physical generator restricts to $b$ times a logical Pauli; after choosing the inner logical basis so that this Pauli is the rotation axis seen by the outer code, a physical phase $\delta\theta$ appears to the outer code as $b\,\delta\theta$. Hence, if the outer code admits an exact logical gate at angle $\theta_0$, the concatenated construction reaches that same gate when
\begin{equation}
  b\,\delta\theta=\theta_0,
  \qquad
  \delta\theta=\frac{\theta_0}{b}.
  \label{eq:ghz-gate-amplification}
\end{equation}
At fixed $\omega_0$ and $\Delta h$, the required hold time is therefore reduced by a factor $b$. This is a resource--protection trade-off rather than a free gain: each outer-code qubit now occupies $b$ physical qubits and the inner GHZ layer is maximally sensitive to phase noise along the amplified generator, which the outer layer must tolerate or correct.

Let $\ket{\overline{\pm}}=(\ket{0}^{\otimes b}\pm\ket{1}^{\otimes b})/\sqrt2$, so $\ket{0}_L=\ket{\overline{-}}^{\otimes N_{\rm blk}}$, $\ket{1}_L=\ket{\overline{+}}^{\otimes N_{\rm blk}}$. Write $\Wop_B=\sum_{j=1}^{N_{\rm blk}}\Wop^{(j)}$ with $\Wop^{(j)}$ the block sum. Since $\Wop^{(j)}\ket{\overline{\pm}}=b\ket{\overline{\mp}}$,
\begin{equation}
  \braket{\overline{\pm}|\Wop^{(j)}|\overline{\pm}}=0,
  \qquad
  \braket{\overline{\pm}|(\Wop^{(j)})^2|\overline{\pm}}=b^2.
\end{equation}
Blocks are uncorrelated on any product of $\ket{\overline{\pm}}$, so $\braket{\Wop_B}=0$ and $\braket{\Wop_B^2}=N_{\rm blk}b^2=bn$. For the logical Bell state $\ket{\Psi_0}=(\ket{0_L}_A\ket{1_L}_B+\ket{1_L}_A\ket{0_L}_B)/\sqrt2$ the cross terms vanish because $\braket{0_L|1_L}_A=0$, so the same moments hold, and with signal generator $G_{\rm sig}\equiv\Wop_B/2$,
\begin{equation}
  \mathcal{F}_Q=4\,\mathrm{Var}(G_{\rm sig})=\braket{\Wop_B^2}-\braket{\Wop_B}^2=b\,n.
\end{equation}
The endpoints $b=1$ and $b=n$ are the phase-flip code ($\mathcal{F}_Q=n$, full $Z$ protection, no amplification) and the GHZ code ($\mathcal{F}_Q=n^2$, no $Z$ protection, maximal amplification). The $b=1$ endpoint is therefore first-order insensitive in the protected code space, not exactly clock-off in the sense of the clock-off definition in the main text: its quantum Fisher information is nonzero and generic finite excursions are filtered rather than stationary.

\paragraph{Clock/no-clock extension.} At an admitted finite-angle logical gate, an eigenspace of the induced logical unitary is invariant up to a common phase and can therefore furnish an endpoint clock-off sector, while an orthogonal sector containing distinct logical eigenphases can be phase-sensitive. This differs from a continuously stationary no-clock state, whose physical state is invariant throughout the hold. A full proper-time-interference protocol, including the relation between these two notions and the nonlocal construction of Ref.~\cite{Borregaard2025}, lies beyond the scope of the present work.

That $n^2$ is a noiseless limit. A $b$-qubit block accumulates signal phase $b\,\delta\theta$ and Gaussian phase-noise variance $b^2\sigma_1^2$, where $\sigma_1^2$ is the single-qubit collective phase-noise variance over an interrogation. Its off-diagonal visibility is therefore $e^{-b^2\sigma_1^2/2}$, so the phase quantum Fisher information of one equatorial block is $b^2e^{-b^2\sigma_1^2}$; with $N_{\rm blk}=n/b$ independent blocks,
\begin{equation}
  \mathcal{F}_Q(\sigma_1)=n\,b\,e^{-b^2\sigma_1^2},
  \qquad
  b^\ast=\frac{1}{\sigma_1\sqrt2},
  \label{eq:noisyfisher}
\end{equation}
where the second expression is the continuous optimum obtained from $\partial_b\ln\mathcal{F}_Q=1/b-2b\sigma_1^2=0$. Relative to $b=1$,
$\mathcal{F}_Q(b^\ast)/\mathcal{F}_Q(1)=e^{\sigma_1^2-1/2}/(\sigma_1\sqrt2)$; hence the constrained optimum uses $b>1$ only for $\sigma_1<1/\sqrt2\simeq0.71$~rad, and for $\sigma_1\ll1$ the gain is $1/(\sigma_1\sqrt{2e})\simeq0.429/\sigma_1$. In an actual code family $b$ is an allowed integer divisor of $n$, so the nearest permitted value is chosen. With $T_2^{\mathrm{cm}}$ the common-mode dephasing time, $\sigma_1^2=2T/T_2^{\mathrm{cm}}$ and $T_\pi\propto1/\Delta h$, a larger separation is doubly favourable---at $T_2^{\mathrm{cm}}=10^3$~s, the continuous optima are $b^\ast=0.4,\,2.1,\,6.6$ at $\Delta h=450$~m, $10$~km and $100$~km---and $\sigma_1$ is the residual after the tracking of Appendix~\ref{app:collective}, not the raw dephasing. Estimating the redshift \emph{rate} with free interrogation time is the unfavourable case. Let $T$ be the interrogation duration and $T_{\mathrm{tot}}$ the total averaging time. With $\sigma_1^2=2T/T_2^{\mathrm{cm}}$, $T_{\mathrm{tot}}/T$ repetitions give
$\mathcal{F}_{Q,\mathrm{tot}}(T)=T_{\mathrm{tot}}nbT\exp[-2b^2T/T_2^{\mathrm{cm}}]$.
Setting $\partial_T\mathcal{F}_{Q,\mathrm{tot}}=0$ gives $T^\ast=T_2^{\mathrm{cm}}/2b^2$ and hence $\mathcal{F}_{Q,\mathrm{tot}}=T_{\mathrm{tot}}nT_2^{\mathrm{cm}}/2eb$, decreasing in $b$, so $b=1$ is optimal there~\cite{Huelga97}. This metrological optimum is distinct from the finite-angle gate use above: for an admitted outer-code gate, choosing $b>1$ can deliberately shorten the required physical excursion according to Eq.~\eqref{eq:ghz-gate-amplification}, provided the resulting resource and noise-sensitivity costs are acceptable.

\paragraph{Relation to the Hamiltonian-not-in-Lindblad-span (HNLS) criterion.}---Reference~\cite{ZZPJ18} shows that Heisenberg scaling is attainable by error correction iff the signal Hamiltonian lies outside the span of $\Id$, the Lindblad operators and their products. Writing $L_\ell$ for the Lindblad operators and $G_{\rm sig}$ for the signal generator, with $P$ the code projector as in the main text, the criterion requires $P L_\ell P\propto P$ together with $P G_{\rm sig}P\neq\text{const}\,P$. For $G_{\rm sig}=\Wop/2$ with $Z$-type jumps those two requirements are exactly the hypotheses of Theorem~\ref{thm:tradeoff}, and $\Wop$ lies in the span, so HNLS fails and no code attains Heisenberg scaling: $\mathcal{F}_Q=bn$ above is a noiseless quantum Fisher information, and the asymptotic scaling in interrogation time $t$ is $O(t)$ rather than $O(t^2)$. Theorem~\ref{thm:tradeoff} is that obstruction in stabilizer form. What differs is the use made of it. Reference~\cite{ZZPJ18} treats $P G_{\rm sig}P\propto P$ as failure; in the present setting it is the design goal for the storage mode, and at $\delta\theta=\pi$ the same suppression returns a deterministic gate---a regime outside an estimation framework linearised in the parameter, and one to which no Fisher information attaches.

% Circuit-level resource modelling omitted from this version; feasibility claims are intentionally structural rather than based on an independent-fault lifetime estimate.

\section{Collective phase: a coherent offset, not a lifetime}\label{app:collective}

Two quantities must be kept apart. Let $\tau_{\mathrm{QEC}}$ be the syndrome period and $p_L$ the logical failure probability per round. The \emph{logical lifetime} $T_L\simeq\tau_{\mathrm{QEC}}/p_L$ decides whether the state survives to $T_\pi$, and is set by the code. With $\Wop_B\equiv\sum_i Z_i$ on register $B$, the relational proper-time unitary is $\Urel=e^{-i\delta\theta\Wop_B/2}$ with $\delta\theta\equiv\omega_0\Delta\tau$, where $\omega_0$ is the proper transition frequency and $\Delta\tau$ the proper-time difference. The \emph{collective phase} $\phi$, entering as $e^{-i\phi\Wop_B/2}$, is not a lifetime at all: by Theorem~\ref{thm:tradeoff} it is the same one-parameter family as $\Urel$, so it displaces the gate phase rather than dephasing the logical qubit. It is a coherent systematic, and the relevant question is not how large it is but how well it is known.

For the phase-flip code, the conditional process fidelity at the deterministic $\delta\theta=\pi$ working point is $F_{\Xbar}^{\rm cond}=s_\delta^6/(c_\delta^6+s_\delta^6)$, with $c_\delta\equiv\cos(\delta\theta/2)$ and $s_\delta\equiv\sin(\delta\theta/2)$. Appendix~\ref{app:survival} derives this expression, the associated survival probability, and the corresponding postselected $\sqrt{\Xbar}$ fidelity about $\delta\theta=\pi/2$. Below, ``surv.'' denotes the code-space survival probability for the maximally mixed logical input. For a \emph{fixed} offset $\epsilon$ in $\delta\theta$:

\begin{center}
\begin{tabular}{ccccc}
\hline\hline
 & \multicolumn{2}{c}{$\delta\theta=\pi$} & \multicolumn{2}{c}{$\delta\theta=\pi/2$} \\
$\epsilon$ (rad) & $1-F_{\Xbar}^{\rm cond}$ & surv. & $1-F_{\sqrt{\Xbar}}^{\rm cond}$ & surv. \\
\hline
$0.01$ & $1.6\times10^{-14}$ & $0.9999$ & $2.2\times10^{-4}$ & $0.250$ \\
$0.10$ & $1.6\times10^{-8}$  & $0.9925$ & $2.2\times10^{-2}$ & $0.258$ \\
$0.1\pi$ & $1.6\times10^{-5}$ & $0.9284$ & $1.7\times10^{-1}$ & $0.322$ \\
$0.50$ & $2.8\times10^{-4}$  & $0.8276$ & $3.0\times10^{-1}$ & $0.422$ \\
$1.00$ & $2.6\times10^{-2}$  & $0.4689$ & $4.7\times10^{-1}$ & $0.781$ \\
\hline\hline
\end{tabular}
\end{center}

The two working points differ in kind. At $\delta\theta=\pi$ the identity amplitude $c_\delta^3$ vanishes, so the conditional process infidelity begins at sixth order: a fixed half-radian offset costs $2.8\times10^{-4}$, while most of the penalty appears as a reduced probability of obtaining the accepted trivial-syndrome outcome. At $\delta\theta=\pi/2$ the two logical amplitudes change oppositely at first order, so the postselected logical rotation angle is linearly sensitive to $\epsilon$ and its conditional infidelity begins at second order, $1-F_{\sqrt{\Xbar}}^{\rm cond}=9\epsilon^2/4+O(\epsilon^4)$ (Appendix~\ref{app:survival}). Let $\sigma_\phi^2$ denote the variance of the residual unmonitored phase; halving the hold halves $\sigma_\phi^2$, but the shorter working point is substantially more phase-sensitive: for a Gaussian residual with $\sigma_\phi=0.5$~rad, numerical averaging gives mean conditional infidelities $6.7\times10^{-3}$ for the $\Xbar$ gate and $0.205$ for the postselected $\sqrt{\Xbar}$ gate. The latter is preferable only when logical lifetime, rather than phase tracking, is the dominant constraint.

What costs is the \emph{unmonitored} part. Writing $\phi=\bar\phi+\tilde\phi$, with $\bar\phi$ the tracked component and $\tilde\phi$ the residual of variance $\sigma_\phi^2$, the quoted averages are obtained from
\begin{equation}
\begin{aligned}
  \langle F_{\Xbar}^{\rm cond}\rangle_{\sigma_\phi}
  &=\frac{1}{\sqrt{2\pi}\sigma_\phi}
    \int_{-\infty}^{\infty}d\epsilon\,
    e^{-\epsilon^2/(2\sigma_\phi^2)}\\
  &\quad\times
    \frac{\cos^6(\epsilon/2)}
         {\sin^6(\epsilon/2)+\cos^6(\epsilon/2)} .
\end{aligned}
\end{equation}
Numerical evaluation gives $\langle F_{\Xbar}^{\rm cond}\rangle=0.9999,\,0.993,\,0.923,\,0.673$ at $\sigma_\phi=0.25,\,0.5,\,0.84,\,1.59$~rad. Comparing the two tables: a random half-radian is $24$ times more costly than a fixed one, the whole difference being the tails. The requirement is therefore on tracking, not on stability: $\sigma_\phi\lesssim0.5$~rad for $\langle F_{\Xbar}^{\rm cond}\rangle>0.99$, corresponding to a fractional frequency residual of $7.8\times10^{-15}$ over $T_\pi=1267$~s.

The sources are known and each is addressed conventionally. On the $m_F=0\!\to\!m_F=0$ clock transition, with $m_F$ the magnetic-sublevel quantum number, the first-order Zeeman shift vanishes, leaving the second-order Zeeman, local-oscillator and blackbody shifts. The first is deterministic given the magnetic field $B_{\rm mag}(t)$ and its sensitivity scales with the bias field, so magnetometry plus operation at low bias converts it from a random to a known offset; the second enters only through the $A$--$B$ relative phase and is tracked against an auxiliary stabilised transfer reference used for calibration rather than to drive the local QEC circuits; the third is a temperature-stability requirement. None is improved by encoding, and none is a limit of the code: the two axes are independent.

Structurally, $\bar\phi$ separates from the signal without needing to be small. The gravitational phase is proportional to $\Delta h$ and $\bar\phi$ is not, so repeating the protocol at several separations isolates the metric contribution as the slope and absorbs $\bar\phi$ into the intercept---the method by which the redshift was established across $33$~cm~\cite{Chou10}, a millimetre~\cite{Bothwell22} and $450$~m~\cite{Takamoto20}.

Two remarks on what does \emph{not} help. First, $\phi$ is per-register: at metre-to-kilometre separations the two modules' field, trap and thermal environments are independent, and the local oscillator is the one correlated component. Second, the entangled state gives no protection. The logical Bell state $\ket{\Psi^+}\equiv(\ket{0_L1_L}+\ket{1_L0_L})/\sqrt2$ is decoherence-free for collective \emph{logical} $\Zbar$, but for the phase-flip code $P\Urel P=c_\delta^3P+i s_\delta^3\Xbar$, so a collective $\Wop$ rotation generates $\Xbar$; $\ket{\Psi^+}$ is an eigenstate of $\Xbar_A\Xbar_B$ but not of $\Xbar_A+\Xbar_B$. The protected axis is therefore the wrong one, which is Theorem~\ref{thm:tradeoff} once more.

\bibliographystyle{apsrev4-2}
\bibliography{relational_invariance_arxiv_v36}

@article{Gundogan2026,
  author = {G{\"u}ndo\u{g}an, Mustafa and Barzel, Roy and R{\"a}tzel, Dennis},
  title = {Gravitational time dilation in quantum clock interferometry with entangled multi-photon states and quantum memories},
  journal = {Quantum Sci. Technol.},
  volume = {11},
  number = {4},
  pages = {045047},
  year = {2026},
  doi = {10.1088/2058-9565/aea2c3}
}

@article{Barzel2024,
  author  = {Barzel, Roy and G{\"u}ndo\u{g}an, Mustafa and Krutzik, Markus and R{\"a}tzel, Dennis and L{\"a}mmerzahl, Claus},
  title   = {Entanglement dynamics of photon pairs and quantum memories in the gravitational field of the earth},
  journal = {Quantum},
  volume  = {8},
  pages   = {1273},
  year    = {2024},
  doi     = {10.22331/q-2024-02-29-1273}
}

@article{Covey2025,
  author  = {Covey, Jacob P. and Pikovski, Igor and Borregaard, Johannes},
  title   = {Probing Curved Spacetime with a Distributed Atomic Processor Clock},
  journal = {PRX Quantum},
  volume  = {6},
  pages   = {030310},
  year    = {2025},
  doi     = {10.1103/q188-b1cr}
}

@article{Yu2026,
  author  = {Yu, Haocun and Macri, Dorotea and Morling, Thomas and Polini, Eleonora and Mieling, Thomas B. and Barrow, Peter and Kabag{\"o}z, Beg{\"u}m and Yin, Xinghui and Chru{\'s}ciel, Piotr T. and Hilweg, Christopher and Oelker, Eric and Mavalvala, Nergis and Walther, Philip},
  title   = {50-km Fiber Interferometer for Testing Gravitational Signatures in Quantum Interference},
  journal = {Phys. Rev. Lett.},
  volume  = {136},
  pages   = {110803},
  year    = {2026},
  doi     = {10.1103/yvlk-t1j9}
}

@article{PoundRebka60,
  author  = {Pound, R. V. and Rebka, G. A.},
  title   = {Apparent Weight of Photons},
  journal = {Phys. Rev. Lett.},
  volume  = {4},
  pages   = {337--341},
  year    = {1960},
  doi     = {10.1103/PhysRevLett.4.337}
}

@article{Chou10,
  author  = {Chou, C. W. and Hume, D. B. and Rosenband, T. and Wineland, D. J.},
  title   = {Optical Clocks and Relativity},
  journal = {Science},
  volume  = {329},
  number  = {5999},
  pages   = {1630--1633},
  year    = {2010},
  doi     = {10.1126/science.1192720}
}

@article{Takamoto20,
  author  = {Takamoto, Masao and Ushijima, Ichiro and Ohmae, Noriaki
             and Yahagi, Toshihiro and Kokado, Kensuke and Shinkai, Hisaaki
             and Katori, Hidetoshi},
  title   = {Test of general relativity by a pair of transportable optical lattice clocks},
  journal = {Nat. Photonics},
  volume  = {14},
  pages   = {411--415},
  year    = {2020},
  doi     = {10.1038/s41566-020-0619-8}
}

@article{Bothwell22,
  author  = {Bothwell, Tobias and Kennedy, Colin J. and Aeppli, Alexander
             and Kedar, Dhruv and Robinson, John M. and Oelker, Eric
             and Staron, Alexander and Ye, Jun},
  title   = {Resolving the gravitational redshift across a millimetre-scale atomic sample},
  journal = {Nature},
  volume  = {602},
  pages   = {420--424},
  year    = {2022},
  doi     = {10.1038/s41586-021-04349-7}
}

@article{ZCPB11,
  author  = {Zych, Magdalena and Costa, Fabio and Pikovski, Igor and Brukner, {\v C}aslav},
  title   = {Quantum interferometric visibility as a witness of general relativistic proper time},
  journal = {Nat. Commun.},
  volume  = {2},
  pages   = {505},
  year    = {2011},
  doi     = {10.1038/ncomms1498}
}

@article{Pikovski15,
  author  = {Pikovski, Igor and Zych, Magdalena and Costa, Fabio and Brukner, {\v C}aslav},
  title   = {Universal decoherence due to gravitational time dilation},
  journal = {Nat. Phys.},
  volume  = {11},
  pages   = {668--672},
  year    = {2015},
  doi     = {10.1038/nphys3366}
}

@article{Loriani19,
  author  = {Loriani, Sina and Friedrich, Alexander and Ufrecht, Christian
             and Di Pumpo, Fabio and Kleinert, Stephan and Abend, Sven
             and Gaaloul, Naceur and Meiners, Christian and Schubert, Christian
             and Tell, Dorothee and Wodey, {\'E}tienne and Zych, Magdalena
             and Ertmer, Wolfgang and Roura, Albert and Schlippert, Dennis
             and Schleich, Wolfgang P. and Rasel, Ernst M. and Giese, Enno},
  title   = {Interference of clocks: A quantum twin paradox},
  journal = {Sci. Adv.},
  volume  = {5},
  number  = {10},
  pages   = {eaax8966},
  year    = {2019},
  doi     = {10.1126/sciadv.aax8966}
}

@article{Roura20,
  author  = {Roura, Albert},
  title   = {Gravitational Redshift in Quantum-Clock Interferometry},
  journal = {Phys. Rev. X},
  volume  = {10},
  pages   = {021014},
  year    = {2020},
  doi     = {10.1103/PhysRevX.10.021014}
}

@article{Borregaard2025,
  author  = {Borregaard, Johannes and Pikovski, Igor},
  title   = {Testing quantum theory on curved spacetime with quantum networks},
  journal = {Phys. Rev. Res.},
  volume  = {7},
  pages   = {023192},
  year    = {2025},
  doi     = {10.1103/PhysRevResearch.7.023192}
}

@article{Sorci26,
  author  = {Sorci, Gabriel and Foo, Joshua and Leibfried, Dietrich and Sanner, Christian and Pikovski, Igor},
  title   = {Quantum Signatures of Proper Time in Optical Ion Clocks},
  journal = {Phys. Rev. Lett.},
  volume  = {136},
  pages   = {163602},
  year    = {2026},
  doi     = {10.1103/qhj9-pc2b}
}

@article{LeMaitre25,
  author  = {LeMaitre, Philip A. and Perche, T. Rick and Krumm, Marius and Briegel, Hans J.},
  title   = {Universal Quantum Computer from Relativistic Motion},
  journal = {Phys. Rev. Lett.},
  volume  = {134},
  pages   = {190601},
  year    = {2025},
  doi     = {10.1103/PhysRevLett.134.190601}
}

@article{Preskill00,
  author       = {Preskill, John},
  title        = {Quantum clock synchronization and quantum error correction},
  journal      = {arXiv preprint},
  eprint       = {quant-ph/0010098},
  archivePrefix= {arXiv},
  year         = {2000}
}

@article{KL97,
  author  = {Knill, Emanuel and Laflamme, Raymond},
  title   = {Theory of quantum error-correcting codes},
  journal = {Phys. Rev. A},
  volume  = {55},
  pages   = {900--911},
  year    = {1997},
  doi     = {10.1103/PhysRevA.55.900}
}

@article{EK09,
  author  = {Eastin, Bryan and Knill, Emanuel},
  title   = {Restrictions on Transversal Encoded Quantum Gate Sets},
  journal = {Phys. Rev. Lett.},
  volume  = {102},
  pages   = {110502},
  year    = {2009},
  doi     = {10.1103/PhysRevLett.102.110502}
}

@article{Faist20,
  author  = {Faist, Philippe and Nezami, Sepehr and Albert, Victor V.
             and Salton, Grant and Pastawski, Fernando and Hayden, Patrick
             and Preskill, John},
  title   = {Continuous Symmetries and Approximate Quantum Error Correction},
  journal = {Phys. Rev. X},
  volume  = {10},
  pages   = {041018},
  year    = {2020},
  doi     = {10.1103/PhysRevX.10.041018}
}

@article{WA20,
  author  = {Woods, Mischa P. and Alhambra, {\'A}lvaro M.},
  title   = {Continuous groups of transversal gates for quantum error correcting codes from finite clock reference frames},
  journal = {Quantum},
  volume  = {4},
  pages   = {245},
  year    = {2020},
  doi     = {10.22331/q-2020-03-23-245}
}

@article{HNPS,
  author  = {Hayden, Patrick and Nezami, Sepehr and Popescu, Sandu and Salton, Grant},
  title   = {Error Correction of Quantum Reference Frame Information},
  journal = {PRX Quantum},
  volume  = {2},
  pages   = {010326},
  year    = {2021},
  doi     = {10.1103/PRXQuantum.2.010326}
}

@article{Huelga97,
  author  = {Huelga, S. F. and Macchiavello, C. and Pellizzari, T.
             and Ekert, A. K. and Plenio, M. B. and Cirac, J. I.},
  title   = {Improvement of Frequency Standards with Quantum Entanglement},
  journal = {Phys. Rev. Lett.},
  volume  = {79},
  pages   = {3865--3868},
  year    = {1997},
  doi     = {10.1103/PhysRevLett.79.3865}
}

@article{ZZPJ18,
  author  = {Zhou, Sisi and Zhang, Mengzhen and Preskill, John and Jiang, Liang},
  title   = {Achieving the Heisenberg limit in quantum metrology using quantum error correction},
  journal = {Nat. Commun.},
  volume  = {9},
  pages   = {78},
  year    = {2018},
  doi     = {10.1038/s41467-017-02510-3}
}

@article{BASESTEP25,
  author  = {Leonhardt, M. and Schweitzer, D. and Abbass, F. and Anjum, K. K. and Arndt, B. and Erlewein, S. and Endoh, S. and Geissler, P. and Imamura, T. and J{\"a}ger, J. I. and Latacz, B. M. and Micke, P. and Voelksen, F. and Yildiz, H. and Blaum, K. and Devlin, J. A. and Matsuda, Y. and Ospelkaus, C. and Quint, W. and Soter, A. and Walz, J. and Yamazaki, Y. and Ulmer, S. and Smorra, C.},
  title   = {Proton transport from the antimatter factory of CERN},
  journal = {Nature},
  volume  = {641},
  pages   = {871--875},
  year    = {2025},
  doi     = {10.1038/s41586-025-08926-y}
}

@article{Diosi89,
  author  = {Di{\'o}si, L.},
  title   = {Models for universal reduction of macroscopic quantum fluctuations},
  journal = {Phys. Rev. A},
  volume  = {40},
  pages   = {1165--1174},
  year    = {1989},
  doi     = {10.1103/PhysRevA.40.1165}
}

@article{Penrose96,
  author  = {Penrose, Roger},
  title   = {On gravity's role in quantum state reduction},
  journal = {Gen. Relativ. Gravit.},
  volume  = {28},
  pages   = {581--600},
  year    = {1996},
  doi     = {10.1007/BF02105068}
}

@article{Bassi13,
  author  = {Bassi, Angelo and Lochan, Kinjalk and Satin, Seema
             and Singh, Tejinder P. and Ulbricht, Hendrik},
  title   = {Models of wave-function collapse, underlying theories, and experimental tests},
  journal = {Rev. Mod. Phys.},
  volume  = {85},
  pages   = {471--527},
  year    = {2013},
  doi     = {10.1103/RevModPhys.85.471}
}
\end{document}